\documentclass[manuscrips,screen]{acmart}

\pdfoutput=1
\AtBeginDocument{%
  }

\acmSubmissionID{328}

\usepackage[]{adjustbox}
\usepackage[]{tikz}

\usepackage[]{hyperref}
\usepackage[utf8]{inputenc} 
\usepackage{hyperref}       
\usepackage{url}            
 \usepackage{booktabs}       
\usepackage{nicefrac}       
\usepackage{microtype}      
\usepackage{xcolor}         

\usepackage[]{booktabs}
\usepackage[capitalise]{cleveref}

\newcommand{\Naive}{\hat{\rm ATE}_{\rm Naive}}
\newcommand{\Qpi}{Q_{\pi_{1 / 2}}}

\newcommand{\Epi}{\E_{\pi_{1 / 2}}}
\newcommand{\Qpihat}{\hat{Q}_{\pi_{1 / 2}}}

\newcommand{\field}[1]{\ensuremath{\mathbb{#1}}}
\newcommand{\R}{\ensuremath{\field{R}}} 
\newcommand{\E}{\ensuremath{\mathsf{E}}} 

\newcommand{\Sscr}{\ensuremath{\mathcal S}}

\newtheorem{assumption}{Assumption}

\newtheorem{theorem}{Theorem}

\newtheorem{corollary}{Corollary}
\newtheorem{lemma}{Lemma}

\begin{document}

\title{Correcting for Interference in Experiments: A Case Study at Douyin}

\author{Vivek F. Farias}
\affiliation{%
  \institution{Massachusetts Institute of Technology}
  \country{USA}
}

\author{Hao Li}
\affiliation{%
  \institution{ByteDance}
  \country{China}}

\author{Tianyi Peng}
\affiliation{%
  \institution{Massachusetts Institute of Technology}
  \country{USA}
}

\author{Xinyuyang Ren}
\affiliation{%
 \institution{ByteDance}
 \country{China}}

\author{Huawei Zhang}
\affiliation{%
  \institution{ByteDance}
  \country{China}}

\author{Andrew Zheng}
\affiliation{%
  \institution{Massachusetts Institute of Technology}
  \country{USA}}


\begin{abstract}
Interference is a ubiquitous problem in experiments conducted on two-sided content marketplaces, such as Douyin (China's analog of TikTok). In many cases, creators are the natural unit of experimentation, but creators interfere with each other through competition for viewers' limited time and attention. ``Naive'' estimators currently used in practice simply ignore the interference, but in doing so incur bias on the order of the treatment effect. We formalize the problem of inference in such experiments as one of policy evaluation. Off-policy estimators, while unbiased, are impractically high variance. We introduce a novel Monte-Carlo estimator, based on ``Differences-in-Qs'' (DQ) techniques, which achieves bias that is second-order in the treatment effect, while remaining sample-efficient to estimate. On the theoretical side, our contribution is to develop a generalized theory of Taylor expansions for policy evaluation, which extends DQ theory to all major MDP formulations. On the practical side, we implement our estimator on Douyin's experimentation platform, and in the process develop DQ into a truly ``plug-and-play'' estimator for interference in real-world settings: one which provides robust, low-bias, low-variance treatment effect estimates; admits computationally cheap, asymptotically exact uncertainty quantification; and reduces MSE by 99\% compared to the best existing alternatives in our applications.

\end{abstract}

%



\maketitle

  \section{Introduction}
\label{sec:intro}


In recent years, large-scale online platforms have increasingly relied on experimentation to measure the impact of interventions and enhance user experiences \citep{siroker2015b,gilotte2018offline,king2017designing}. Traditional randomized controlled trials (RCTs), commonly known as A/B tests, offer a robust framework for causal inference in settings where treatment and control groups are independent \citep{stolberg2004randomized,deaton2018understanding}. However, in many online platforms, treatment and control outcomes are not independent but interact through a shared system state, a phenomenon known as "Markovian" interference \citep{farias2022markovian}. Addressing this interference is a key challenge for state-of-the-art experimentation platforms \citep{kohavi2020trustworthy,blake2014marketplace,holtz2018limiting,li2022interference}.

In this work, we propose a novel method to address problems of interference that arise in online content marketplaces such as Douyin, a leading social video platform with 600 million daily active users as of early 2021.\footnote{\url{https://new.qq.com/rain/a/20210329A06EP300}} Like its US analog TikTok, Douyin's core product is short-form video content: creators create videos on the app, and viewers are presented with a sequence of these videos (determined by the company's proprietary recommendation algorithm). In a typical viewer session, the viewer opens the app; the platform presents them with a video; the viewer watches this video until they swipe to indicate they would like to advance to the next video; and the platform then presents them with another video. The process repeats until the viewer eventually leaves the platform.

As with many other two-sided marketplaces, a key challenge in experimentation at Douyin is that certain interventions can {\it only} be tested via creator-side experiments, while many key metrics are measured on the viewer side. To give a specific, real example: Douyin has a new feature which will allow creators to deliver livestreams in high definition, and Douyin wants to test the impact of this new feature on viewer engagement. A typical metric will be average ``dwell time'' for a viewer -- that is, the total time an average viewer spends on the platform in some interval.

Because the feature is introduced at the creator level, the platform cannot control which viewers will be exposed, as would be necessary for viewer-side or two-side randomized experiments. Interference between creators occurs primarily because viewers have a limited budget of resources to spend on Douyin -- a budget consisting of free time, attention, battery life, and other factors -- and creators in treatment compete for this budget with those in control. See \cref{fig:tiktok-graphic} for an illustrative example. In such settings, the only existing option deployed at Douyin is naive estimation; in other words, to ignore the interference and proceed with estimation as in a standard RCT. However, this naive estimator is known to be considerably biased, with the bias potentially being as large as the treatment effect itself \citep{blake2014marketplace,holtzLimitingBiasTestControl2020}.

\vspace{0.5em}
\noindent\textbf{Contributions.}
To address the interference problem, we first formulate treatment effect estimation as a problem of off-policy evaluation in a Markovian Decision Process (MDP), where states can be arbitrarily complex -- consisting for example of a viewer's preferences, engagement level, and viewing history. Existing off-policy evaluation methods, despite being unbiased, fail in this setting due to the prohibitively large state space, and are excessively high variance. In contrast, the naive estimator does not require access to the state, but it suffers from a significant bias. To overcome these limitations, we propose a novel estimator that does not require state access and has a second-order bias compared to the naive estimator. Our proposed estimator can be viewed as an instance of Differences-in-Qs (DQ) estimators, proposed in a very recent theoretical work \citep{farias2022markovian} which showed that DQ estimators enjoy a favorable bias-variance tradeoff in average-reward MDPs.

We integrate our proposed DQ estimator into the Douyin platform, incorporating various generalizations, and demonstrate its superior performance compared to existing methods. To summarize, our contributions are threefold:

\textbf{1) A Novel and Practical Estimator.} We propose a novel DQ estimator to correct for interference in A/B tests at Douyin. This estimator does not require access to the state, making it more broadly applicable compared to other off-policy evaluation methods. It has provably second-order bias (\cref{th:total-reward-bias}) and offers practical benefits, as it accommodates heterogeneous users, partially observable states, and multiple simultaneous experiments (\cref{sec:tiktok-extensions}). Additionally, we develop doubly robust techniques to reduce variance (\cref{sec:doubly-robust}) and enable a precise variance quantification (\cref{sec:permutation}).

\textbf{2) Superior Empirical Performance.} We implement our estimator on Douyin's experimentation platform, and report results from a large-scale simulator modeling this implementation in \cref{sec:tiktok-experiments}. These experiments show that the DQ estimator significantly outperforms state-of-the-art approaches, reducing mean squared error (MSE) by 99\% compared to the best existing alternatives. Our work represents the first large-scale implementation of the DQ estimator in the real world.

\textbf{3) A Unified Theory for DQ.} Finally, we provide theoretical extensions to the DQ estimator and its underlying Taylor series theory (\cref{sec:unified-theory}), allowing it to accommodate the discounted and total reward formulations more naturally suited to Douyin's applications. Our unifying theory yields existing average reward results in \cite{farias2022markovian} as a special case, and are of independent interest.
\vspace{-0.5em}
\section{Related Work}
Interference in experiments, across fields such as public medicine, agriculture, and online markets \citep{struchiner1990behaviour,talbot1995effect,lucking1999using}, occurs when the outcome of an experimental unit is affected by others, potentially leading to biases in treatment effect estimation. Interference has emerged as a key challenge in online platforms in particular, across companies including eBay \citep{blake2014marketplace}, Uber/Lyft \citep{chamandyExperimentationRidesharingMarketplace2016,uber2019experimentation}, Airbnb \citep{holtzLimitingBiasTestControl2020}, and Douyin, the setting which motivates this study.

\textbf{Experimental Design.} Existing approaches to interference have primarily focused on sophisticated experimental designs. Examples include minimizing interference through clustering units \citep{vaver2011measuring, walker2014design, pouget2019variance, zigler2021bipartite}, alternating randomization across time periods \citep{bojinov2019time, hohnhold2015focus, bojinov2022design}, conducting randomization on both supply and demand sides in two-sided markets \citep{johari2022experimental, li2022interference, bajari2021multiple}, and other designs \citep{glynn2020adaptive, tchetgen2012causal, johari2017peeking}. Despite their potential, practical limitations such as cost and implementation concerns often restrict the use of such sophisticated designs \citep{kirnChallengesExperimentation2022, kohavi2020trustworthy}, as we will also see in motivating applications at Douyin. In contrast, as opposed to introducing more complex designs, this paper introduces a novel estimator under \textit{simple unit-level randomization}, effectively mitigating bias induced by interference.

\textbf{Off-Policy Evaluation (OPE).} Our problem can be viewed as a case of \emph{off-policy evaluation (OPE)} \citep{precup2001off, sutton2008convergent}, an increasingly important problem in reinforcement learning. Unbiased OPE estimators typically suffer from high variance \citep{thomas2015high, thomas2016data, liu2018breaking}, and the related curse of dimensionality in large state spaces \citep{jiang2016doubly, kallus2020double, kallus2022efficiently}. Our approach can be construed as a {\it biased} method for OPE, which incurs a small amount of bias for a massive reduction in variance (similarly to \citep{schulman2015trust, farias2022markovian}) and offers strong theoretical guarantees without the need to access state information.

  \section{Model}
\label{sec:tiktok-model}

We begin by developing our theory in a simplified setting, where we observe a single trajectory (or ``session'') from \textit{a single viewer}. In \cref{sec:tiktok-extensions}, we will extend this core theory to a much richer problem formulation, capturing the complexities of applications at Douyin.

Consider a scenario in which the platform wants to test a creator-side intervention -- for example, rolling out the HD streaming feature discussed in the introduction. Due to infrastructure constraints or concerns about the impact of the intervention, we must experiment at a creator level, and therefore implement a creator-side A/B test.


We model each session for the viewer as an MDP. At each step $t$, a video is presented to the viewer. With probability $p$, the presented video is in the treatment group, corresponding to an action $a_{t} = 1$; otherwise the video is in the control group, corresponding to an action $a_{t} = 0$. When the viewer swipes to the next video, the platform realizes a reward $r_{t}$ -- e.g., the time the viewer spent watching the previous video -- and the MDP advances to the next step. The reward $r_{t}$ is a (random) function of $a_{t} \in \mathcal{A}$, the action at time $t$, and $s_{t} \in \mathcal{S}$, the viewer's state at time $t$.

Critically, in real applications, the state may be arbitrarily complex, and partially or totally unobservable. For example, $s_{t}$ may consist of the time the viewer has spent on the platform (and implicitly, the remaining time budget of the viewer), their engagement level, and even the whole history of videos the viewer has watched. See \cref{fig:tiktok-graphic} for an example. As we will see, we can design estimators for which this complexity poses no challenge.

\begin{figure}[htbp]
  \centering
  \includegraphics[scale=0.8]{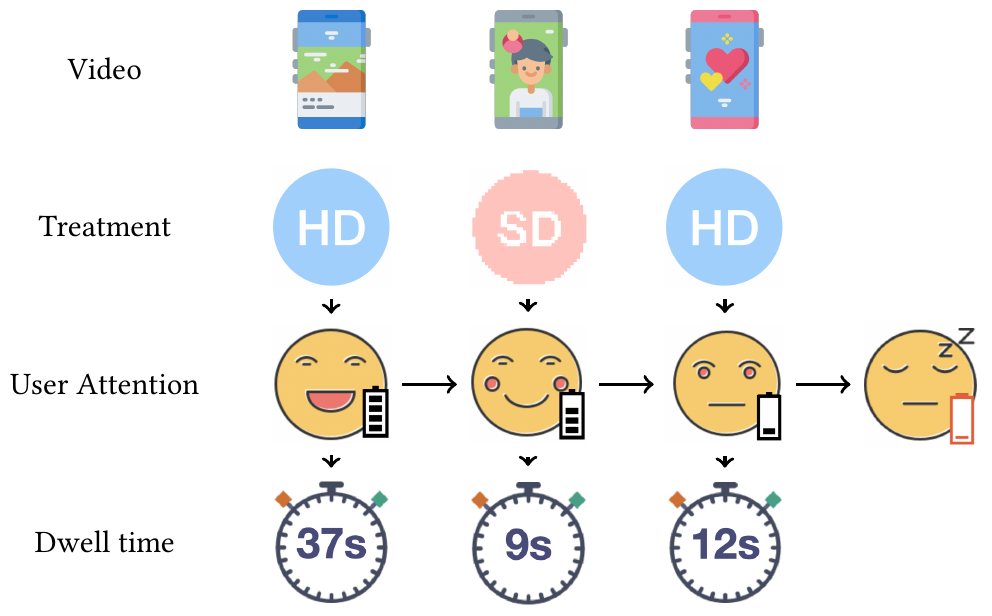}
  \caption{A graphical depiction of the Douyin MDP for a given viewer. At each timestep, the viewer is shown a video; the video is either treatment (high definition, $a_t = 1$) or control (standard definition, $a_t = 0$); and depending on the video and the viewer's current state (consisting of e.g., attention, battery life, data budget, etc.), the platform then collects some reward (i.e., the viewer's dwell time), and moves to the next video. Eventually the viewer runs out of attention and leaves the platform.}
  \label{fig:tiktok-graphic}
\end{figure}

A policy $\pi : \mathcal{S} \rightarrow \mathcal{A}$ is a (random) mapping from states to actions. We are interested in three policies:
\begin{enumerate}
\item {\bf Global treatment} $\pi_{1}$, which chooses $\pi_{1}(s)=1, \forall s \in \mathcal{S}$, i.e., every video presented is treated. The associated transition matrix is denoted as $P_{1} \in \R^{|\mathcal{S}|\times |\mathcal{S}|}$.
\item {\bf Global control} $\pi_{0}$, which chooses $\pi_{0}(s) = 0, \forall s \in \mathcal{S}$, i.e., every video is untreated. The associated transition matrix is denoted as $P_{0} \in \R^{|\mathcal{S}|\times |\mathcal{S}|}$.
\item {\bf The experiment} $\pi_{p}$, which chooses $\pi(s) = 1$ with probability $p$; otherwise $\pi(s)=0$; in other words, video-side A/B testing.  The associated transition matrix is denoted as $P_{p} := (1-p)\cdot P_{0}+ p \cdot P_{1}$. For ease of exposition, we set $p=1/2$ and discuss generalizations in \cref{sec:tiktok-extensions}.
\end{enumerate}
For a given policy $\pi$, we define the total reward of the policy as: $J_{\pi} = \sum_{t=0}^{\infty} \E_{\pi}[r_t ]$, where the notation $\E_{\pi}$ denotes an expectation over states $s_0 \sim \rho_{\rm init}$, actions $a_t | s_t \sim \pi(s_t)$, dynamics $s_{t+1} | s_t \sim P_{\pi}(s_t, s_{t+1})$, and rewards.

\textbf{Goal.}  Given observations $\{(r_{t}, a_{t})\}$ under the experiment policy $\pi_{1/2}$, our goal is to estimate the difference between the total rewards obtained by $\pi_{1}$ and $\pi_0$:
\begin{equation}
  \mathrm{ATE} := J_{\pi_{1}} - J_{\pi_0}.
\end{equation}
In our example above, $\mathrm{ATE}$ corresponds to the impact to the viewer's dwell time in the platform by replacing the policy $\pi_0$ by $\pi_1$, a key quantity of interest when deciding whether rolling out a new policy.

To ensure that $J_{\pi}$ is well-defined, we make one key assumption on the Markov chain induced by each policy, which is that the viewer eventually terminates their session; i.e., all Markov chains are {\it absorbing}. To be precise: we say that $\mathcal{S}_{\rm abs} \subseteq \mathcal{S}$ is an absorbing class under policy $\pi$ if $P_{\pi}(s, s') = \mathbb{I}\left\{s' \in \mathcal{S}_{\rm abs}\right\}$ for all $s \in \mathcal{S}_{\rm abs}$, and if $r_t = 0$ a.s. when $s_t \in \mathcal{S}_{\mathrm{abs}}$. We also require that $\mathcal{S}_{\mathrm{abs}}$ is reached in finite time almost surely starting from any state and define the expected time to this absorption event as $T_{\rm abs}^{\pi} := \max_{s \in \Sscr} \sum_{t=0}^{\infty} \E_{\pi}[\mathbb{I}\left\{s_t \notin \mathcal{S}_{\rm abs}\right\} | s_0 = s]$. We then make the following assumption:
\begin{assumption}
  \label{ass:expt-absorbing}
  There exists a class $\mathcal{S}_{\rm abs} \subseteq \mathcal{S}$ and a time $T_{\rm abs}$, such that $\mathcal{S}_{\rm abs}$ is absorbing and $T_{\rm abs} > T_{\mathrm{abs}}^{\pi}$ for each policy $\pi$ in $\pi_0, \pi_{1}, \pi_{1 / 2}$.
\end{assumption}

  \section{A Novel Estimator for Estimation Under Interference}
\label{sec:tiktok-estimator}

\subsection{A Naive Estimator}
We begin by describing the status quo for this setting: Naive estimation, wherein interference is simply ignored and the experiment is treated as a traditional A/B test. A typical estimator is the inverse propensity weighted estimator
\begin{align}
\label{eq:naive}
\hat{\mathrm{ATE}}_{\rm Naive}  := \sum_{t=0}^{\infty} \left(\frac{1(a_{t}=1)}{\Pr(a_{t}=1)} r_{t} -  \frac{1(a_{t}=0)}{\Pr(a_{t}=0)} r_{t}\right) = \sum_{t=0}^{\infty} 2\cdot \left(1(a_{t}=1) r_{t} - 1(a_{t}=0) r_{t}\right),
\end{align}
which simply estimates the difference between the single-step rewards obtained under treatment and control -- ignoring the impact of each action on the overall trajectory.

However, the bias of this estimator can be significant. To see this, note that in expectation
\begin{align}
\label{eq:naive-expectation}
 \E_{\pi_{1 / 2}} \left[ \hat{\mathrm{ATE}}_{\rm Naive} \right] =   \sum_{t=0}^{\infty} \E_{\pi_{1/2}} \left[ r(s_t, 1) \right] - \E_{\pi_{1/2}}\left[r(s_{t}, 0) \right]
\end{align}
On the other hand, $\mathrm{ATE} = \sum_{t=0}^{\infty} \E_{\pi_1} [ r(s_t, 1)] -  \sum_{t=0}^{\infty} \E_{\pi_0} [ r(s_t, 0)].$ The difference between $\mathrm{ATE}$ and $ \Epi \left[ \hat{\mathrm{ATE}}_{\rm Naive} \right]$ then results from the fact that the expectation we compute by sampling from $\pi_{1 / 2}$ does not match the expectations we would like to compute, under $\pi_{0}, \pi_{1}$.

Intuitively, then, the magnitude of the bias depends on the difference between the dynamics under $\pi_{0}$ and under $\pi_{1}$. We characterize this difference as follows: starting from any state, let $\delta$ bound the distance between distributions over next states. Precisely, $\delta := \max_{s} D_{\mathrm{TV}}(P_0(s, \cdot), P_1(s, \cdot))$ where $D_{\mathrm{TV}}$ measures the total variation distance\footnote{Note that in a typical A/B testing, $P_{0}$ is often close to $P_{1}$ so that $\delta$ is expected to be small.}. Then, the bias of $\Naive$ is $\Theta(\delta)$, whereas the ATE itself is also $\Theta(\delta)$ -- i.e., the bias of Naive estimation is on the same order as the treatment effect. We will demonstrate this in a number of ways: \cref{sec:examples} gives a simple example where the Naive estimator measures an effect which is in fact only cannibalization; \cref{sec:tiktok-experiments} contains realistic experiments in which the bias of Naive is at least 900\% of the treatment effect; and \cref{sec:Taylor-DQ} provides a rigorous bound to this effect.

\subsection{The Differences-in-Qs Estimator}
Intuitively, Naive's bias results from its myopia: Naive simply ignores the impact of the treatment on future rewards. This inspires us to propose a less myopic estimator, which we show here for the special case $p=1/2$\footnote{This greatly simplifies exposition; we discuss the general case in \cref{sec:tiktok-extensions}.}:
\begin{align}
\label{eq:dq}
\hat{\mathrm{ATE}}_{\rm DQ}  := \sum_{t=0}^{\infty} 2\cdot \left(1(a_{t}=1) \left(\sum_{t'=t}^{\infty} r_{t'}\right) -  1(a_{t}=0) \left(\sum_{t'=t}^{\infty} r_{t'}\right) \right),
\end{align}

In other words, we simply replace the rewards $r_{t}$ in \cref{eq:naive} by {\it long-term} accumulated rewards. We refer to this as a "Differences-in-Qs" (DQ) estimator, so called because it estimates the quantity
\begin{equation}
  \label{eq:dq-expectation}
 \Epi \left[ \hat{\mathrm{ATE}}_{\rm DQ} \right] =   \sum_{t=0}^{\infty} \Epi \left[ \Qpi(s_{t}, 1) - \Qpi(s_{t}, 0) \right]
\end{equation}
where $Q_{\pi}(s, a):= \sum_{t=0}^{\infty} E_{\pi}[r_{t}|s_{0}=s, a_{0}=a]$ is the usual state-action value function in MDPs; i.e., the accumulated long-term rewards starting from the state $s$ and the action $a$. This is actually an instantiation of a broad class of DQ estimators, including prior work \cite{farias2022markovian} as special cases, which we discuss in \cref{sec:unified-theory}.

Surprisingly, this simple change yields a dramatic improvement in the bias in estimating $\mathrm{ATE}$. In fact, the bias of \cref{eq:dq} is {\it second-order} in $\delta$; precisely, we have

\begin{theorem}[Bias of DQ]
\label{th:total-reward-bias}
Assume that $\forall s \in \Sscr$, $D_{\rm TV}( P_1(s,\cdot), P_0(s,\cdot) ) \leq \delta$, and let $r_{\max} := \max_{s,a} \left|r(s,a)\right|$. Then,
\[
\left|{\rm ATE} - \Epi\left[
\hat {\rm ATE}_{\rm DQ}
\right]\right| \leq T_{\rm abs}^{3} r_{\max} \cdot \delta^2.
\]
\end{theorem}

Reducing bias from $\delta$ to $\delta^2$ produces huge gains in practice: \cref{sec:examples} gives a simple setting in which DQ removes bias entirely, and empirically the bias of DQ is negligible, as we show in \cref{sec:tiktok-experiments}.%

Below we remark on some salient properties of the estimator.

\textbf{DQ is agnostic to state.} Beyond its low bias, the greatest advantage of $\hat{\mathrm{ATE}}_{\rm DQ}$ is its simplicity, in particular the fact that it does not require observations of state $s_{t}$ to implement. This is critical when $s_{t}$ is partially observed or extremely high dimensional, as in many real applications -- in the Douyin setting state consists of a viewer's sentiment, attention, preferences, device status, and any number of other factors, which are only partially observable by proxy and may be totally distinct from session to session. This lack of state also avoids any assumptions on the functional form of $\Qpi$ and avoids any bias introduced by using function approximation or state aggregation to estimate $\Qpi$.

\textbf{Credit Assignment Interpretation.} This estimator also has a surprising and intuitive interpretation as an explicit ``credit assignment'' mechanism. Rewriting \cref{eq:dq-tot-mc} explicitly, we have $\hat{\mathrm{ATE}}_{\rm DQ} = \sum_{t=0}^{\infty}
    2\cdot \left(\#_{t}(1) - \#_{t}(0) \right) r_{t}$
where $\#_{t}(a) = \sum_{t' = 0}^{t} \mathbb{I}\left\{a_{t'} = a\right\}$ is the count of actions $a$ played up to and including time $t$. In effect: we are simply reweighting each reward $r_{t}$ in an intuitive way, by assigning credit to each action according to the relative contribution of that action in realizing $r_{t}$.

\subsection{Examples}
\label{sec:examples}
Finally, we provide some specific examples to build intuition on each of these estimators.

\noindent\textbf{Example 1. (Extreme interference).} Consider a scenario where the viewer will watch control videos for 15 minutes each, and treatment videos for 20 minutes each; i.e., $r(s_{t}, 1) = 20$ min and $r(s_{t}, 0)=15$ min. However, the viewer has a fixed attention budget of 30 minutes, after which they leave the platform, stopping in the middle of a video if necessary. Since the total dwell time is fixed, the true ATE is 0. Upon inspection, we also see that $Q(s,1) = Q(s,0)$ for any state, since the total future dwell time does not depend on actions. As a result, DQ is unbiased:
\begin{align*}
 \E_{\pi_{1/2}} \left[ \hat{\mathrm{ATE}}_{\rm DQ} \right] =   \sum_{t=0}^{\infty} \E_{\pi_{1/2}} \left[ Q(s_{t}, 1) - Q(s_{t}, 0) \right] = 0. 
\end{align*}
On the other hand, one can compute that $\Epi \left[ \hat{\mathrm{ATE}}_{\rm Naive} \right] = \sum_{t=0}^{\infty} \Epi \left[ r(s_t, 1) - r(s_{t}, 0) \right] = 2.5 \text{min}.$ This discrepancy results essentially from myopia: without accounting for the downstream effects (or lack thereof) of each action, the effect measured by Naive turns out to be fully a result of cannibalization.

\noindent \textbf{Example 2. (No interference).} Again, let $r(s_{t}, 1) = 20$ min and $r(s_{t}, 0)=15$ min. Now suppose that a viewer will watch three videos a day, independent of the treatments. There is no interference in this case: each action impacts watch time for a single video, and has no impact on total videos watched; and thus the total effect of treatment is the sum of the single-step effects. In this case ${\mathrm{ATE}} = (20-15)\cdot 3=15$min. The naive estimator is clearly unbiased, while the DQ estimator is unbiased as well in this scenario since for any $t$,
\begin{align*}
 \Epi\left[1(a_{t}=1) \left(\sum_{t'=t}^{\infty} r_{t'}\right) - 1(a_{t}=0)\left(\sum_{t'=t}^{\infty} r_{t'}\right) \right] = \Epi\left[1(a_{t}=1) r_{t} -  1(a_{t}=0) r_{t} \right],
\end{align*}
i.e., $a_{t}$ does not impact the values of $\Epi[r_{t'}]$ for $t'>t.$

  \section{A Unified Theory for DQ}
\label{sec:unified-theory}
While the previous section develops a high-level intuition for DQ, here we provide a rigorous analysis which frames DQ as a first-order approximation to the ATE. This interpretation immediately yields the bias bound \cref{th:total-reward-bias}, as well as a variety of generalizations of the estimator.

More generally, whereas existing theory covers specific reward formulations (specifically average reward \citep{farias2022markovian}), one of our main theoretical contributions is to introduce a {\it unified} framework for Taylor series expansions in all major reward formulations. These include finite horizon with total reward; infinite horizon with discounted reward; total reward in absorbing Markov chains (which models the Douyin application); and average reward in ergodic Markov chains \citep{farias2022markovian}. Our analysis reveals the fundamental role played by the `effective horizon' in various scenarios (see \cref{tab:settings}). In addition, our unifying theory of DQ covers \citep{farias2022markovian} as a special case by adopting a distinct, yet more straightforward proof technique that can be of independent interest. 

\subsection{Background}

We begin by defining precisely each of the Markovian reward settings we address: discounted, average, and total reward settings. In all settings, we consider policies $\pi, \pi^{\prime}$ inducing transition matrices $P_{\pi}, P_{\pi'}$, and reward functions $r_{\pi}, r_{\pi'} : \mathcal{S} \mapsto \mathbb{R}$ where $r_{\pi}(s) = \E_{a \sim \pi(s)}[r(s, a)]$. For each setting and policy, we will define a ``value functional'' $\tilde{\E}_{\pi}$, such that $\tilde{\E}_{\pi} r_{\pi}$ is the objective of interest (either the discounted reward, average reward, or total reward for policy $\pi$) and a corresponding $Q$ function (the `reward-to-go' after taking some action at a certain state); and we make assumptions on $\pi, \pi'$ to ensure that these values are well-defined.

\begin{assumption}[Discounted reward]
  \label{ass:discounted}
  Let $\gamma \in [0, 1)$. For any reward function $r$, define $\tilde{\E}_{\pi} r = \sum_{t=0}^\infty \gamma^{t}\E_{\pi}[r(s_t)]$ and $Q_{\pi}(s, a; r) = \sum_{t=0}^\infty \gamma^{t}\E_{\pi}[r(s_t) | s_0 = s, a_0 = a].$
\end{assumption}

\begin{assumption}[Average reward]
  \label{ass:average}
  Let $P_{\pi}, P_{\pi'}$ be ergodic Markov chains with stationary distributions $\rho_{\pi}, \rho_{\pi'}$, and mixing rate $\max_{s} D_{\rm TV}(P_{\pi}^{k}(s, \cdot), \rho_{\pi}) \leq C \beta^{k}, \forall k \in \mathbb{N}$ for some constants $C$ and $\beta<1$. Let $\tilde{\E}_{\pi} r = \lim_{T \to \infty} \frac{1}{T}\sum_{t=0}^{T-1} \E_{\pi}[r(s_t)]$ and $Q_{\pi}(s, a; r) = \lim_{T\rightarrow \infty} \sum_{t=0}^{T-1} \E_{\pi}[r(s_t) - \tilde{E}_{\pi} r| s_0 = s, a_0 = a].$
\end{assumption}

\begin{assumption}[Finite horizon reward]
  \label{ass:finite}
  Let $H \in \mathbb{N}$ be the horizon, and let $\tilde{\E}_{\pi} r = \sum_{t=0}^{H-1} \E_{\pi}[r(s_t)]$ and $Q_{\pi}(s_t, a_t; r) = \sum_{t'=t}^{H-1} \E_{\pi}[r(s_t') | s_t = s, a_t = a].$
\end{assumption}

\begin{assumption}[Total reward in absorbing MDPs]
  \label{ass:absorbing}
  Let $P_{\pi}, P_{\pi'}$ be Markov chains with absorbing states $\mathcal{S}_{\rm abs} \subseteq \mathcal{S}$ and time to absorption bounded by $T_{\rm abs}$. Let $\tilde{\E}_{\pi} r = \sum_{t=0}^{\infty } \E_{\pi}[r(s_t)]$ and $Q_{\pi}(s, a; r) = \sum_{t=0}^{\infty} \E_{\pi}[r(s_t) | s_0 = s, a_0 = a]$. This corresponds to the setting in \cref{sec:tiktok-model}.
\end{assumption}

\subsection{A Unified Taylor Series Framework}

We can now state our main theorem:
\begin{theorem}
  \label{th:taylor-fundamental}
Let $\pi, \pi'$ satisfy any one of Assumptions~\ref{ass:discounted}, \ref{ass:average}, \ref{ass:finite}, \ref{ass:absorbing}.
   \begin{equation*}
     \tilde{E}_{\pi'} r = \left(\sum_{k=0}^{K} \tilde{E}_{\pi} [\mathrm{DQ}_{\pi, \pi'}^{(k)}(s)]\right)+ \tilde{E}_{\pi'} [\mathrm{DQ}_{\pi, \pi'}^{(K+1)}(s)]
   \end{equation*}
   where ${\rm DQ}_{\pi, \pi'}^{(0)}(s) = r_{\pi'}(s)$ and ${\rm DQ}_{\pi,\pi'}^{(k)}(s) = \E_{a \sim \pi'(s)}[Q_{\pi}(s, a; {\rm DQ}_{\pi,\pi'}^{(k-1)})] - \E_{a \sim \pi(s)}[Q_{\pi}(s, a; {\rm DQ}_{\pi,\pi'}^{(k-1)})]$.
    \end{theorem}

    To interpret: \cref{th:taylor-fundamental} provides a $K$-th order approximation to evaluating the policy $\pi'$, using only trajectories generated by $\pi$ -- simply by taking a ``sum''\footnote{Appropriately weighted and normalized.} of $DQ_{\pi, \pi'}^{(k)}(s)$ terms along such trajectories. The theorem also gives an exact remainder to this approximation, which we bound as $O(\delta^{K+1})$ in the next result:
    \begin{corollary}[$K^{\rm th}$ order remainder]
      \label{co:approx-error}
      Let $\pi, \pi'$ satisfy any one of Assumptions~\ref{ass:discounted}, \ref{ass:average}, \ref{ass:finite}, \ref{ass:absorbing}. Further assume that $\max_{s} D_{\mathrm{TV}}(P_{\pi}(s, \cdot), P_{\pi'}(s, \cdot)) \leq \delta$.  Then,
      \begin{equation*}
     \left\lvert \tilde{\E}_{\pi'} r - \sum_{k=0}^{K} \tilde{\E}_{\pi} [\mathrm{DQ}_{\pi, \pi^{\prime}}^{(k)}(s)] \right\rvert \leq (\delta H_{\rm eff})^{K+1} \|\tilde{\E}_{\pi'}\|_{1}r_{\rm max}
      \end{equation*}
      where the effective horizon $H_{\rm eff}$ and scaling constant $\|\tilde{\E}_{\pi'}\|_{1}$ are defined in \cref{tab:settings}, and $r_{\rm max} = \max_{s \in \mathcal{S}} |r(s)|$.
      \end{corollary}

     We note that the scaling constant simply measures how large the value $\tilde{\E}_{\pi}r$ can be, relative to $r_{\max}$.
The combination of \cref{th:taylor-fundamental} and \cref{co:approx-error} offers an elegant, unified theory for DQ estimators. In the following section, we demonstrate how \cref{th:total-reward-bias} can be derived as a specific instance of this theory, and we reserve the proof for \cref{sec:proof-main-theorem}.
    \begin{table}[h]
\begin{tabular}{r|cccc}
    Reward                            & $\tilde{E}_{\pi} r$                                       & $H_{\rm eff}$              & Scaling Const. $\|\tilde{\E}_{\pi'}\|_{1}$ \\
    \hline
Discounted (\ref{ass:discounted})     & $\sum_{t=0}^{\infty} \gamma^{t} \E_{\pi}[r(s_t)]$ & $1 / (1 - \gamma)$              & $1 / (1 - \gamma)$\\
    Average (\ref{ass:average})       & $\lim_{T \to \infty} \frac{1}{T}\sum_{t=0}^{T-1} \E_{\pi}[r(s_t)]$  & $(2 \log C + 1) / (1 - \beta)$ & 1\\
    Finite Horizon (\ref{ass:finite}) & $\sum_{t=0}^{H-1} \E_{\pi}[r(s_t)]$               & $H$                        & $H$\\
    Absorbing (\ref{ass:absorbing})   & $\sum_{t=0}^{\infty } \E_{\pi}[r(s_t)]$      & $T_{\rm abs}$              & $T_{\mathrm{abs}}$
  \end{tabular}
  \caption{Summary of the reward settings in Assumptions~\ref{ass:discounted}, \ref{ass:average}, \ref{ass:finite}, \ref{ass:absorbing}, along with the constants appearing in \cref{co:approx-error}.}
  \label{tab:settings}
    \end{table}
\subsection{From the Taylor expansion to DQ}\label{sec:Taylor-DQ}
Note that \cref{co:approx-error} suggests a recipe for obtaining estimators of the ATE with bias $O(\delta^2)$: we simply estimate the first-order expansion of $J_{\pi_{1}}= \tilde{\E}_{\pi_{1}} r_{\pi_{1}}$, and subtract from it the same approximation of $J_{\pi_0}= \tilde{\E}_{\pi_0}r_{\pi_0} $. The bias results immediately arises from bounding the remainder. For illustrative purposes, we will now execute this recipe in the context of total reward Markovian interference setting described in \cref{sec:tiktok-model}, which implements \cref{ass:absorbing}.

As a warm-up: consider the zeroth-order expansion of \cref{th:taylor-fundamental} (i.e., we set $K = 0$ and omit the remainder). Then, from definitions, we immediately have
\begin{align*}
  \tilde{\E}_{\pi_{1}}[{\rm DQ}_{\pi_{1 / 2}, \pi_1}^{(0)}(s)] -
  \tilde{\E}_{\pi_{0}}[{\rm DQ}_{\pi_{1 / 2}, \pi_0}^{(0)}(s)] = \sum_{t=0}^{\infty} \E_{\pi_{1 / 2}}[r_{\pi_{1}}(s_t) - r_{\pi_0}(s_t)] = \E_{\pi_{1 / 2}}\hat{\mathrm{ATE}}_{\mathrm{Naive}}.
\end{align*}

In other words: the Naive estimator simply estimates the zeroth-order expansion of \cref{th:taylor-fundamental}. \cref{co:approx-error} then immediately yields the $O(\delta)$ bound on bias in general: $\left\lvert\mathrm{ATE} - \hat{\mathrm{ATE}}_{\rm Naive} \right\rvert \leq  \delta T_{\rm abs}^2 r_{\rm max}$.

To reduce bias, we naturally turn to estimating higher order terms. Continuing with this recipe with $K= 1$, we can show via some algebra that in fact
\begin{align*}
  \tilde{\E}_{\pi_{1}}[{\rm DQ}_{\pi_{1 / 2}, \pi_1}^{(0)}(s) + {\rm DQ}_{\pi_{1 / 2}, \pi_1}^{(1)}(s)]
  - \tilde{\E}_{\pi_{0}}[{\rm DQ}_{\pi_{1 / 2}, \pi_0}^{(0)}(s) + {\rm DQ}_{\pi_{1 / 2}, \pi_0}^{(1)}(s)]
  &= \E_{\pi_{1 / 2}}[\hat{\rm ATE}_{\rm DQ}]
\end{align*}


In other words: the DQ estimator simply estimates the first-order expansion of \cref{th:taylor-fundamental}, and as a result the $O(\delta^2)$ bias of \cref{th:total-reward-bias} follows immediately from \cref{co:approx-error}. Thus, we view DQ as the first-order correction to the Naive estimator; its bias properties, and a recipe for DQ estimators in other settings -- including other reward formulations, experiments other than uniform randomization $p=1 / 2$, and in the general problem of off-policy evaluation for arbitrary policies -- follow immediately from this interpretation. Higher-order estimators can also be derived through similar means, although one pays a cost in variance to estimate them.

\subsection{Proof of \cref{th:taylor-fundamental}}\label{sec:proof-main-theorem}

All cases of \cref{th:taylor-fundamental} have elementary proofs, which follow a unified framework. For intuition, consider the discounted reward as defined in \cref{ass:discounted}, $\tilde{E}_{\pi} r$. From the definition of $P_{\pi}$, it is well known that
\begin{equation*}
\tilde{E}_{\pi}r = \sum_{t= 0}^{\infty } \gamma^{t} P_{\pi}^{t} r =\rho_{\rm init}^\top(I - \gamma P_{\pi})^{-1} r
\end{equation*}
where we overload notation slightly to let $r \in \mathbb{R}^{|\mathcal{S}|}$ be the vector induced by the function $r$. Similar forms exist for each of the other settings, replacing $\gamma P_{\pi}$ with an analogous matrix. To understand how $\tilde{E}_{\pi} r$ depends on $\pi$, then, we must understand how the inverse $(I - \gamma P_{\pi})^{-1}$ varies with $\pi$.

So motivated, we start with the following elementary perturbation bound (which we note, applies to $\gamma P_{\pi}$)
\begin{lemma}
  \label{lem:perturbation-fundamental}
  Let $A, A' \in \mathbb{R}^{n \times n}$ be matrices such that $(I - A)^{-1}$ and $(I - A')^{-1}$ exist. Then,
  \begin{equation}\label{eq:perturbation}
(I - A')^{-1} = (I-A)^{-1}  + (I - A')^{-1}(A' - A) (I - A)^{-1}
  \end{equation}
\end{lemma}
\begin{proof} Observe that $I - A = I - A' + A' - A$. Right-multiplying by $(I-A)^{-1}$ and left-multiplying by $(I-A')^{-1}$, we obtain the results. 
\end{proof}
This immediately leads to a series expansion for the matrix $(I - A)^{-1}$:\footnote{To see this as a ``Taylor'' expansion: suppose we let $A' = A + \delta D$ for some direction $D$. Then, \cref{th:taylor-matrix} immediately yields a Taylor series of the function $f(\delta) = (I-A')^{-1}$ in $\delta$, around $\delta=0$.}
\begin{lemma}
  \label{th:taylor-matrix}
  Let $A, A' \in \mathbb{R}^{n \times n}$ be matrices such that $(I - A)^{-1}$ and $(I - A')^{-1}$ exist.  Then,
  \begin{equation*}
(I - A')^{-1} = (I-A)^{-1} \sum_{k=0}^{K} \left[ (A' - A) (I - A)^{-1}
\right]^{k} + (I - A')^{-1} \left[(A' - A) (I - A)^{-1} \right]^{K+1}
  \end{equation*}
\end{lemma}
\begin{proof}
The proof follows simply from applying \cref{lem:perturbation-fundamental} to $(I- A')^{-1}$ repeatedly up to $K$ times.
\end{proof}

  In the discounted case, \cref{th:taylor-fundamental} then follows immediately from \cref{th:taylor-matrix} by letting $A = \gamma P_{\pi}$ and $A' = \gamma P_{\pi'}$. The other total reward cases follow from the same proof by choosing appropriate analogs of $\gamma P_{\pi}$; whereas the average reward case can be shown as a limiting regime where $\gamma \to 1$. We provide these proofs in full in the Appendix.

  \newcommand{\ATEDQ}{\hat{\rm ATE}_{\rm DQ}}

\section{Implementing DQ at Douyin}
\label{sec:tiktok-extensions}

Here, we expand the model to allow for much more of the complexity present in real applications, in particular in content marketplaces like Douyin/TikTok. In particular, the real world introduces two main challenges not yet discussed: first, as opposed to estimating the treatment effect on a single viewer, we are interested in the aggregate effect across all viewers; and second, as opposed to assigning treatments independently for each viewer for each video, treatments are instead assigned once for each {\it creator}, and this assignment propagates to all viewers.

First, we show that the estimators and theory from \cref{sec:tiktok-model} continue to work in this setting, with slight modification. Second, we will augment these estimators with additional techniques to fulfill practical requirements: variance reduction, tight variance estimation, and fast implementations. Ultimately, the estimator we describe will be one that, out of the box, provides highly accurate and low variance treatment effect estimates under interference, that satisfy the constraints of real applications.

\subsection{A Richer Model of Content Platforms}
We begin by detailing the precise model we consider. As mentioned, the total reward formulation of \cref{sec:tiktok-model} models the trajectory of a single viewer session, whereas at the platform level, we are interested in an aggregate effect across all viewers. Note also that each viewer session may be completely heterogeneous. In other words, each viewer session $i \in [N]$ may have a completely distinct reward function $r^{(i)}$, transition kernel $P^{(i)}$, and total reward $J_{\pi}^{(i)}$. To model this, we will assume that $r, P, J$ are drawn iid from some unknown distribution $\mathcal{M}$ over MDPs; and for each such draw, we observe an entire trajectory $\{s_{it}, a_{it}, r_{it}\}_{t \geq 0}$. Our problem is then, given $N$ such trajectories (i.e., viewer sessions), to estimate the expected ATE under $\mathcal{M}$:
\begin{equation}
  \label{eq:ate-meta}
  \overline{\rm ATE} = \E_{\mathcal{M}}[J_{\pi_{1}} - J_{\pi_{0}}]
\end{equation}
Second, the treatment assignment occurs not independently across viewers and timesteps, but rather once for each creator; and any viewer who encounters that creator, each time they encounter that creator, will receive the same treatment. More formally, there exists a population of creators indexed by $j \in [M]$, each with a random treatment assignment $a(j) \sim \pi_{p}$. Viewer $i$ at timestep $t$ is shown video $j_{it}$, and the action played at that timestep is then $a_{it} = a(j_{it})$.

Finally, data collection is typically much more complicated in the real world -- there may be multiple treatment groups, and the randomization probability $p$ is typically much smaller than $1 / 2$. Casting the problem as one of off-policy evaluation allows us to extend the framework naturally to such cases -- the data collection policy can essentially be arbitrary. We defer the discussion to \cref{sec:data-collection}.


\subsection{Monte-Carlo Estimation At Douyin}
To be precise, we begin by discussing a simple Monte-Carlo estimator which can be applied to this more general problem. Recall from \cref{sec:tiktok-estimator} that this the estimator \cref{eq:dq} is for a single trajectory. This then suggests a straightforward extension to estimate the meta-treatment effect \cref{eq:ate-meta}, by simply aggregating DQ estimators across trajectories. Let $\hat{\rm ATE}_{\rm DQ}^{(i)} = \sum_{t=0}^{\infty} \left[\hat{Q}^{\rm MC}_{\pi_{1 / 2}}(s_{it}, 1) - \hat{Q}^{\rm MC}_{\pi_{1 / 2}}(s_{it}, 0) \right]$ be the Monte-Carlo DQ estimator for trajectory $i$, where $\hat{Q}^{\rm MC}_{\pi_{1 / 2}}(s_{it}, a) = 2\cdot \mathbb{I}\left\{a_{it} = a \right\} \sum_{t=0}^{\infty} r_{it}$ is the Monte-Carlo Q-function estimate. We then have the aggregate estimator:
\begin{equation}
  \label{eq:meta-dq-mc}
  \hat{\mathrm{ATE}}_{\rm DQ} = \frac{1}{N} \sum_{i=1}^{N} \hat{\rm ATE}_{\rm DQ}^{(i)}.
\end{equation}
By linearity of expectation, the bias of this estimator will be upper bounded by the expected value of the bias bound in \cref{th:total-reward-bias}, where the expectation is taken over MDPs drawn from $\mathcal{M}$.

Regarding the issue of creator-level treatment assignments: As long as the treatment does not alter a video's probability of being shown to the user (i.e., the recommendation system), then from a single viewer's perspective, the marginal treatment probabilities $\Pr\left(a_{it} = a\right)$ remain unchanged, and as a result the bias properties of our various estimators remain unchanged. However, this correlation of treatment assignments across viewers introduces challenges for variance quantification, which we address in \cref{sec:permutation}.\footnote{In fact, a single viewer may view the same video multiple times in the same session, creating dependence in action probabilities across time. In the Appendix, we give a general version of our estimator which accounts for this dependence correctly.}


\subsection{A Doubly Robust DQ Estimator}
\label{sec:doubly-robust}

The Monte-Carlo estimator \cref{eq:meta-dq-mc} is essentially free of assumptions, and is trivial to implement -- it requires only reward observations, and can be computed in a single pass through the dataset. However, this flexibility comes at a cost in terms of variance, as we show in our experiments.

Motivated by the empirical success of doubly robust estimators \cite{jiang2016doubly}, we introduce a doubly robust DQ estimator. Suppose we have an approximation to the $Q$ function, $\hat{Q}^{\rm Reg}$ (typically obtained via regression on pre-experiment data, or on some held out set of viewers, as we do in \cref{sec:tiktok-experiments}), which may be misspecified or biased. Then, we define the estimator
\begin{equation}
  \label{eq:dq-tot-mc}
  \hat{\mathrm{ATE}}_{\rm DQ-DR} = \frac{1}{N} \sum_{i=1}^{N}\sum_{t=0}^{\infty} \left[\hat{Q}_{\pi_{1 / 2}}^{\rm DR}(s_{it}, 1) - \hat{Q}_{\pi_{1 / 2}}^{\rm DR}(s_{it}, 0) \right]
\end{equation}
where we now introduce an approximate $Q$-function and use instead the doubly robust Q-function estimator:
\begin{equation*}
\Qpihat^{\rm DR}(s_{it}, a) = \Qpihat^{\rm Reg}(s_{it}, a) + \frac{\mathbb{I}\left\{a_{it} = a\right\}}{\pi_{1/2}(s_{it}, a)} \left[ \sum_{t'=t}^{\infty}  r_{it'} - \Qpihat^{\rm Reg}(s_{it}, a) \right]
\end{equation*}
where $\pi(s_{it}, a)$ is the probability of selecting action $a$ under the policy $\pi$ and state $s_{it}.$ This estimator has three salient properties. First, $\Qpihat^{\rm DR}$ remains an unbiased estimator of $Q_{\pi_{1/2}}$ regardless of how poorly $\Qpihat^{\rm Reg}$ estimates $\Qpi$, as long as the randomization probabilities $\pi(s, a)$  are known. This fact is key, as in the case of partially observable state, $\Qpihat^{\mathrm{Reg}}$ may be fit heuristically using the observable portion of state, and may in addition perform function approximation or regularization. Second, the introduction of $\Qpihat^{\rm Reg}$ serves as a control variate, typically decreasing the variance of the estimator overall; in the experiments we show that this reduces variance by 95\% relative to Monte-Carlo DQ estimation.

Finally, we consider a subtler issue: robustness to the randomization probabilities $\pi(s, a)$. In particular, in many real situations, only an approximation of the randomization probabilities $\hat{\pi}(s, a)$ is known. This arises most often in observational settings, where $\pi(s, a)$ is simply unknown a priori and must be estimated from data; but is also common in real-world experimental settings, where the complexity of implementation can often result in bugs that skew the realized randomization probability, a discrepancy referred to as ``sample ratio mismatch'' in \cite{kohavi2020trustworthy} (which also provides a wide range of examples under which this can occur). As we will see, the Monte-Carlo estimator \cref{eq:meta-dq-mc} is only unbiased when $\hat{\pi} = \pi$, and when this does not hold the bias can be extremely large in practice.

Doubly robust estimators address this issue by sharply reducing the dependence of the estimator on $\pi(s, a)$. For intuition, consider only the estimator $\Qpihat^{\rm DR}$ described above, but replacing $\pi(s,a)$ with an estimate $\hat{\pi}$. The expectation $\Epi \Qpihat^{\mathrm{DR}}$  is then
\begin{equation*}
  \E_{\pi_{1 / 2}} \Qpihat^{\mathrm{DR}}(s, a) =  \frac{\pi(s,a)}{\hat{\pi}(s, a)} [\Qpi(s, a) - \Qpihat^{\rm Reg}(s, a)] + \Qpihat^{\mathrm{Reg}}(s, a)
\end{equation*}
We can then immediately see that the better $\Qpihat^{\mathrm{Reg}}$ estimates $\Qpi$, the less sensitive $\E_{\pi_{1 / 2}} \Qpihat^{\mathrm{DR}}$ will be to $\hat{\pi}$. At one extreme, if $\Qpi = \Qpihat^{\mathrm{Reg}}$, then the doubly robust estimate will be completely unbiased regardless of $\hat{\pi}$; this (together with the fact that $\Qpihat^{\rm DR}$ is also unbiased if $\hat{\pi} = \pi$) is the well-known ``doubly'' robust property of such estimators.

In our experiments \cref{sec:tiktok-experiments}, we show that that the introduction of this doubly robust technique yields dramatically more precise estimators, both when $\pi$ is known exactly and up to perturbation.

\subsection{Variance Characterization under Randomization Testing}
\label{sec:permutation}
A key challenge in understanding the variance of $\hat{\rm ATE}_{\mathrm{DQ}}$ is that actions are assigned at the streamer level; streamers are shared across viewers; and as result $a_{it}$ can be arbitrarily correlated across viewers.

Whereas exact variance characterizations are difficult in this situation, we propose the use of rerandomization tests to test the sharp null hypothesis $H_0$ that the treatment has identically zero effect on the reward distribution or problem dynamics. One typical approach to characterizing the distribution of the test statistic $\hat{\mathrm{ATE}}_{\mathrm{DQ}}$ under $H_0$ is via simulation, where we randomly redraw the assignments $a(j)$ (the assignment of the streamer $j$), and re-compute the test statistic under the new treatment assignments (holding the outcomes constant, as would be prescribed under the null). This ``re-randomization'' approach yields an exact distribution of the test statistic, but is typically infeasible at Douyin's scale: with on the order of hundreds of millions of viewers, a SQL query for a single realization of the treatment assignments $a(j)$ took around 12 hours to run on Douyin's cluster.


As it turns out, the variance of $\hat{\mathrm{ATE}}_{\mathrm{DQ}}$ under the null hypothesis can actually be computed in {\it linear} time. The Monte-Carlo estimator \cref{eq:dq-tot-mc} can actually be written as a sum over streamer-level terms, which eventually leads to an {\it exact} characterization of the permutation testing variance under $H_0$. To see this, note that the estimator \cref{eq:dq-tot-mc} is equivalent to
\begin{equation*}
  \hat{{\rm ATE}}_{\rm DQ} = \sum_{j=1}^{M} \left(\frac{\mathbb{I}\left\{a(j) = 1\right\}}{\Pr(a(j) = 1)} - \frac{\mathbb{I}\left\{a(j) = 0\right\}}{\Pr(a(j) = 0)}\right) \left(\frac{1}{N} \sum_{i=1}^{N}\sum_{t=t_{ij}}^{\infty} r_{it} \right)
\end{equation*}
where $t_{ij}$ is the time period in which viewer $i$ viewed streamer $j$. Each of these summands is independent since each assignment $a(j)$ is independent from each other, and as a result we need only to characterize the variance of each summand and sum them in order to compute the variance of ${\rm ATE}_{\rm DQ}$. To conclude this line of thought, note that each summand is an affine function of a Bernoulli random variable, which has closed form expression for variance:
\begin{equation*}
  {\rm Var}(\hat{\rm ATE}_{\rm DQ}) = \sum_{j=1}^{M} \Pr(a(j) = 0) \Pr(a(j)=1) \left(\frac{1}{\Pr(a(j)=1)} - \frac{1}{\Pr(a(j)=0)}\right)^2 \left(\frac{1}{N} \sum_{i=1}^{N}\sum_{t=t_{ij}}^{\infty} r_{it} \right)^2
  \end{equation*}
%

Using this variance, we can then perform a hypothesis test, either assuming normality (reasonable in practice) or a Chebyshev-type bound. In the experiments, we show that this test is very high-powered in practice, detecting effect sizes of $0.2\%$ in a day, and achieves nearly nominal coverage.

While this gives a flavor of why linear-time permutation testing should be possible, analogous results for $p \neq \frac{1}{2}$ are much more involved. In those cases, naively pursuing the same approach has complexity $O(M^{4})$; while a smarter, exact approach is $O(M^{2})$ and only an approximation (albeit empirically a very good one) is possible in $O(M)$ time. We provide details of this implementation in the Appendix, and show in \cref{sec:tiktok-experiments} that this approach empirically provides nearly exact coverage.

\subsection{General Data Collection Policies}\label{sec:data-collection}

Finally, we provide some examples of how to generalize DQ estimation to variety of data collection scenarios. For exposition we have discussed DQ estimation in the specific setting of treatment effect estimation from a A/B test with treatment probability $p=1 / 2$, which yields the cleanest results. Even modifying $p$, however, results in a non-intuitive generalization. In full generality, we collect data under some data-collecting (or, commonly, ``behavioral'') policy $\pi_{\rm data}$, and we would like to evaluate some new target policy $\pi_{\rm new}$ (or, in some cases, a set of target policies). A range of common scenarios fall under this umbrella:
\vspace{-0.5em}
\begin{itemize}
  \item A/B testing with $p \neq 1 / 2$. Let $\pi_{\rm data} = \mathrm{Bern}(p)$, and we treat $\pi_0, \pi_{1}$ individually as target policies.
  \item RCTs with multiple treatment groups $1 \ldots K$. Let $\pi_{\rm data}(s)$ be a distribution over $1 \ldots K$, and we have target policies $\pi_{1} \ldots \pi_{K}$.
  \item Evaluating a new policy from logged data. Let $\pi_{\rm data}$ be the logging policy, and let $\pi_{\rm new}$ be an arbitrary new policy (with the same support over $\mathcal{A}$ in each state as $\pi_{\rm data}$).
\end{itemize}

On top of this, we also deal with the issue of actions possibly being dependent across time. That is, defining the history $\mathcal{H}_t = \{s_{t'}, a_{t'}, r_{t'}\}_{t' < t}$ we now allow $\pi_{\rm data}, \pi_{\rm new }$ to depend arbitrarily on this history. We denote this dependence as $\pi(s, a | \mathcal{H}_t)$. The Taylor expansion in \cref{th:taylor-fundamental} provides a DQ analog for each of these generalizations. To be precise: the idealized first-order expansion becomes:
\begin{equation*}
J_{\pi_{\rm new}}^{\mathrm{DQ}} = \sum_{t=0}^{\infty} \E_{\pi_{\rm data}} \left[ \E_{a \sim \pi_{\rm new}}[r(s_t, a_t) + Q_{\pi_{\rm data}}(s_t,a_t; r_{\pi_{\rm new}})] - Q_{\pi_{\rm data}}(s_t, a_t; r_{\pi_{\rm new}})\right]
\end{equation*}
Notably, this resulting Monte-Carlo estimator is more complicated to compute, with multiple importance sampling weights, and this introduces further challenges (and opportunities) for each of the previous extensions to DQ. Our implementation at Douyin overcomes these challenges and we defer details to Appendix.
%


  \section{Experiments}
\label{sec:tiktok-experiments}
We now provide empirical results comparing the DQ estimator (and various improvements) on simulations motivated by our particular application at Douyin. Our experimental results aim to address the following research questions:
\begin{itemize}
  \item {\bf RQ1.} Does DQ provide more accurate treatment effect estimates than existing alternatives?
  \item {\bf RQ2.} Do our hypothesis testing approaches have correct coverage and sufficient power to detect small treatment effects at realistic timescales?
  \item {\bf RQ3.} How robust are these results under perturbations to $\pi$?
\end{itemize}


\subsection{Experimental setting}
A typical source of interference at Douyin is in experiments where, due to implementation constraints, Douyin must randomize at the creator (or streamer) level, rather than the viewer level. We take the streaming setting discussed in the introduction as a concrete example: Douyin has a new feature which will allow creators to deliver livestreams in high definition, but as a creator-facing feature, experiments with this feature must be conducted at the creator level. In this setting, HD streaming is almost guaranteed to increase viewer watch time on a per-video basis, but the effect on total watch time is likely to be overestimated by Naive estimation due to cannibalization of viewer attention. Furthermore, the ATE  can plausibly be negative because the feature will tend to consume viewers' battery and data budgets much more quickly.


While we have tested the DQ estimator on internal simulators at Douyin, as well as in real-world experimental settings, here we primarily present results from a simulation calibrated to Douyin user data. This approach ensures privacy protection while maintaining the consistency of the qualitative outcomes. We represent state as $s_{it} = (w_t, u_{i}, v_{it})$, where $w_t$ represents the amount of time the viewer has spent watching videos so far, $u_{i} \in \mathbb{R}^{d}$  (here we choose $d = 5$) is a latent vector representing the viewer's preferences, and $v_{it} \in \mathbb{R}^{d}$ is a latent vector representing the video shown at time $t$ to viewer $i$, such that $\left\langle u_{i}, v_{it} \right\rangle$ measures the affinity of viewer $i$ for the $t^{\rm th}$ video shown. We assume that $u, v$ are invisible to the agent, although in reality the agent may have access to some noisy version of these vectors.

The viewer's counterfactual watch time under control is then generated as $w_{it} \sim \mathrm{Exp}(1 / (k u_{i}^\top v_{it}))$, for some scalar problem parameter $k$. Finally, the viewer's actual watch time is $r_{it}= w_{it} (1 + \tau^{*} a_{it})$; i.e. it increases by a multiplicative factor of $\tau^*$ under treatment $a_{it}=1$. The viewer's state is either their cumulative watch time $s_t = \sum_{t'=1}^{t-1} r_{it'}$, or the terminal state (denoted by $s_{\rm abs}$) which indicates that the viewer has left the platform. At each epoch the viewer has a probability of transitioning  to the terminal state $s_{\rm abs}$ (i.e., ``leaving the platform'') determined by $s_t$ as $\mathbb{P}\left(s_{t+1} = s_{\rm abs} | s_t\right) = \frac{\exp(s_t)}{\alpha + \exp(s_t)}$ where $\alpha$ is a problem parameter. We calibrated the parameters $k, \alpha$ to reflect the actual distribution of watch durations and number of videos watched in Douyin's viewer population; we omit the actual values here.

\subsection{Algorithms}
We will compare the performance of both Monte-Carlo and Doubly Robust DQ variants. For Doubly Robust DQ, we construct a simple regression estimator for $\Qpihat^{\rm Reg}$: we hold out the first 1000 viewers observed under the experiment policy, and fit via least squares a regression model of the form $\hat{Q}_{\pi_{1 / 2}}^{\rm Reg}((w, u, v), a) = \beta_0 + \beta w$ with respect to the parameters $\beta_0, \beta$.

As mentioned, the only alternative being applied in practice is Naive estimation. Although not used in practice, unbiased OPE is possible via importance sampling at the level of the entire trajectory, although as one might expect the variance of such an approach is astronomical. Note that each of these approaches work for partially observable state, and are compatible with the streamer-level randomization required in our scenario. To be precise, the typical Naive estimator is $ \hat{\mathrm{ATE}}_{\rm Naive} = \frac{1}{N} \sum_{i=1}^{N} \sum_{t=0}^{\infty} \left[\hat{r}_{\rm IS}(s_{it}, 1) - \hat{r}_{\rm IS}(s_{it}, 0) \right]$
where $\hat{r}_{\rm IS}$ is the importance sampling estimator $\hat{r}_{\mathrm{IS}}(s_{it}, a) = \frac{\mathbb{I}\left\{a_{it} = a\right\}}{\pi(s_{it}, a)} r_{it}$
and the typical, stepwise importance sampling OPE estimator is
\begin{align*}
  \hat{\mathrm{ATE}}_{\rm OPE} = \frac{1}{N} \sum_{i=1}^{N} \sum_{t=0}^{\infty} \left[\hat{r}_{\rm IS-OPE}(s_{it}, 1) - \hat{r}_{\rm IS-OPE}(s_{it}, 0) \right]
\end{align*}
where $\hat{r}_{\rm IS-OPE}(s_{it}, a_{it})$ is an importance sampling estimator for the expected reward in the $t^{\mathrm{th}}$ period under policy $\pi_{a}$, $\E_{\pi_{a}}[r(s_{it},a_{it})]$. Here we implement the ``stepwise'' importance sampling estimator \cite{jiang2015doubly}, which improves slightly over naive trajectory-level sampling: $\hat{r}_{\mathrm{IS-OPE}}(s_{it}, a) = \left(\prod_{t'=0}^{t} \frac{\mathbb{I}\left\{a_{it'} = a \right\}}{\pi(s_{it'}, a)}\right) r_{it}$.

The doubly robust approaches used to construct \cref{eq:dq-tot-mc} can also be applied to refine the Naive and OPE estimators to provide a more refined baselines. We implemented the state-of-the-art methods from \citep{kallus2020double}. It is worth noting, however, that the DR variants will retain the same bias properties; and therefore one will expect Naive-DR to also exhibit large bias, and for OPE-DR to remain unbiased but to still have unreasonably high variance. We also experiment with these estimators, but defer their derivation to the Appendix.

\subsection{Overall Performance Comparison (RQ1)}

\paragraph{DQ provides much more accurate estimates of ATE than any alternative}

We first consider the root mean-squared error of each estimator in \cref{fig:tiktok-rmse}, as well as a breakdown into bias and variance in \cref{tab:tiktok-distr}. First, we note that interference leads to heavy bias in the Naive estimator: the sign of the Naive estimator is always wrong, a manifestation of the intuition that the intervention is always myopically positive, but can be negative in the long run. Not only is the sign wrong, but the magnitude of the error is large: the Naive estimate is around -800\% of the actual ATE in all cases. Next, we turn to the OPE variants. OPE is indeed unbiased, but the variance is so large that its RMSE is by far the largest out of any of the estimators considered, at all timescales considered. This holds for both doubly robust and vanilla versions of the estimator. Neither Naive nor OPE algorithms achieve relative RMSE below 100\%.

Finally turning to DQ, we see that Monte-Carlo DQ achieves much lower RMSE than any Naive or OPE variant, for sufficiently large effect sizes and sufficient observations $N$. However, its variance remains large even at $10^{8}$ viewers, primarily due to the limited number of streamers observed, and it also fails to achieve relative RMSE below 100\% on the effect sizes measured. Doubly Robust DQ, on the other hand, provides striking performance gains: error is reduced by at least 97\% compared to Monte-Carlo DQ in all instances, and by up to 99\% compared to Naive estimation. By $N = 10^{8}$ (on the order of one day's worth of viewer activity), DQ-DR already achieves relative RMSE below 100\% of the treatment effect on all treatment effects of magnitude greater than 0.17\%.


\begin{figure}
  \centering
  \begin{adjustbox}{width=\textwidth}
   \input{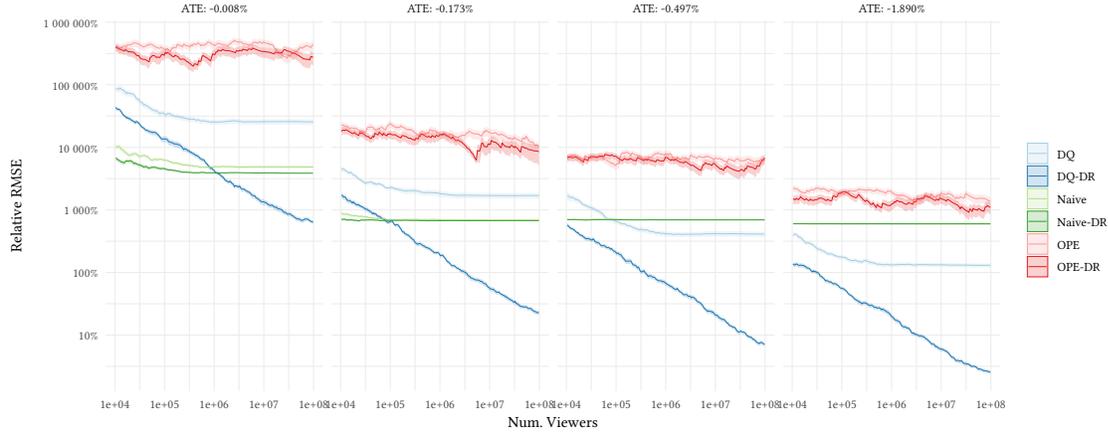}
  \end{adjustbox}
  \caption{Relative RMSE of each estimator vs. $N$, i.e., the number of viewers observed, for various effect sizes. Relative RMSE is measured as $\sqrt{\sum_{i=1}^{T} (\hat{\mathrm{ATE}_i} - \mathrm{ATE})^2} / (\sqrt{T} \mathrm{ATE})$ where $\mathrm{ATE}$ is the true treatment effect and $\hat{\mathrm{ATE}}_{i}$ is the estimator's output for the $i$-th experiment over $T=100$ seeds. Error bars indicate standard errors over $T$ seeds.  See \cref{tab:tiktok-distr} for a breakdown of bias vs. variance for all estimators.
  }
  \label{fig:tiktok-rmse}
\end{figure}
\begin{table}
  \centering
\begin{tabular}[t]{rrrrrrrrr}
\toprule
Actual ATE (\%) & -0.01        &       & -0.17         &       & -0.50        &       & -1.89        &  \\
& Mean         & SD    & Mean         & SD    & Mean         & SD    & Mean         & SD           \\
\midrule
DQ         &  -0.05 (0.00) & 2.01  & -0.15 (0.00) & 2.95  & -0.49 (0.00) & 2.03  & -1.92 (0.00) & 2.44         \\
DQ-DR      & -0.05 (0.00) & 0.03  & -0.15 (0.00) & 0.03  & -0.49 (0.00) & 0.03  & -1.92 (0.00) & 0.04         \\
Naive      & 0.30 (0.00)  & 0.17  & 0.99 (0.00)  & 0.25  & 2.95 (0.00)  & 0.17  & 9.43 (0.00)  & 0.22         \\
Naive-DR   & 0.30 (0.00)  & 0.01  & 0.99 (0.00)  & 0.01  & 2.95 (0.00)  & 0.01  & 9.43 (0.00)  & 0.01         \\
OPE        & -0.01 (0.00) & 35.42 & -0.17 (0.00)  & 17.24 & -0.50 (0.00) & 32.01 & -1.89 (0.00)  & 26.24        \\
OPE-DR     & -0.01 (0.00)  & 22.05 & -0.17 (0.00)  & 14.82 & -0.50 (0.00) & 33.32 & -1.89 (0.00) & 20.87        \\
\bottomrule
\end{tabular}
  \caption{Bias and variance of each estimator, at $N=10^{8}$ viewers observed, as a \% of the baseline outcome $J_{\pi_0}$, for problem instances with various effect sizes (i.e., $-0.01\%$, $-0.17\%$, $-0.50\%$, and $-1.89\%$). Parenthetical quantities are standard errors of the mean estimation. Doubly Robust DQ achieves nearly zero bias and reduces standard deviation of DQ by about 97\%.}
  \label{tab:tiktok-distr}
\end{table}

\subsection{Power and Coverage of Hypothesis Tests (RQ2) }

\paragraph{Rerandomization testing detects effect sizes of 0.2\% using a day of data} We now consider the problem of hypothesis testing under DQ estimation. \cref{fig:tiktok-power} illustrates of the effectiveness of our rerandomization tests as outlined in \cref{sec:permutation}, which generalizes straightforwardly to Doubly Robust DQ. When analyzing a sample size of $10^{8}$ viewers, which is on the same order as about one day of traffic on Douyin, the Doubly Robust DQ estimator demonstrates its robustness by consistently achieving 80\% power for detecting effect sizes as small as 0.15\%. In addition, the test successfully adheres to the target false positive rate of 10\% when the true treatment effect is zero.
\begin{figure}[htbp]
  \centering
    \includegraphics[width=0.8\linewidth]{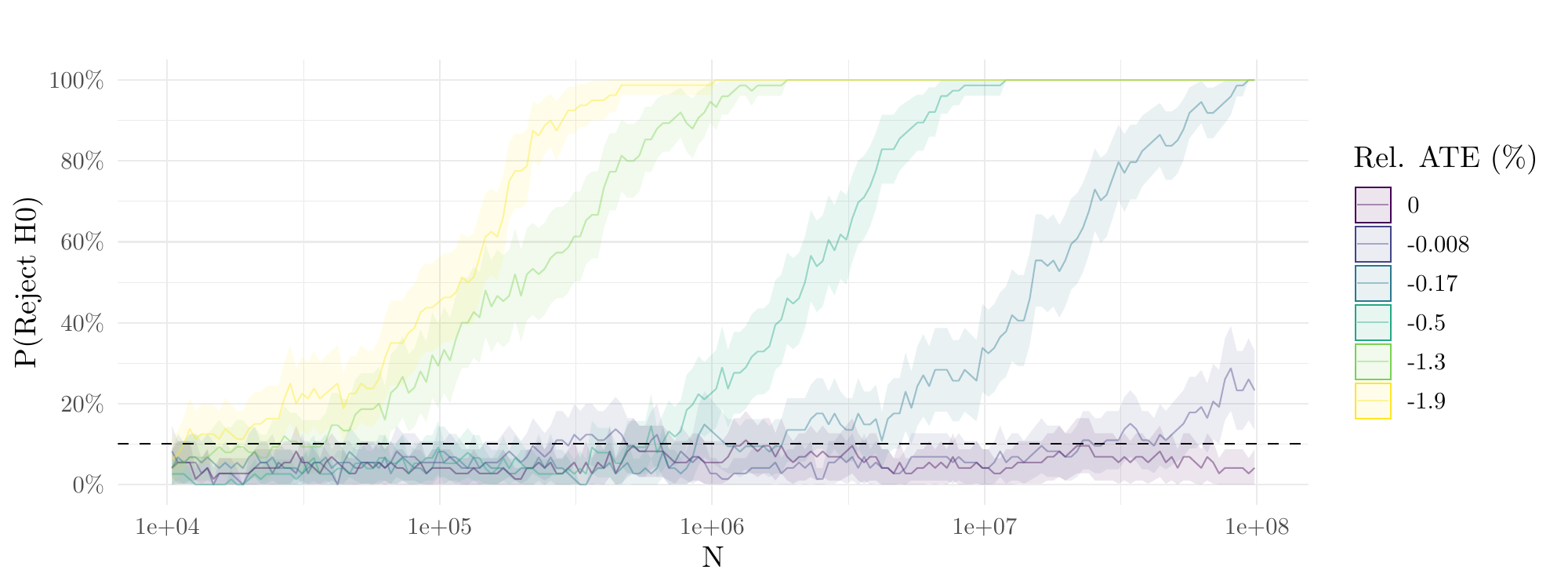}
    \caption{$\mathbb{P}(\text{Reject } H_0)$ (i.e., power) vs. number of viewers sampled for the proposed approximate permutation test, for the doubly robust estimator, for various effect sizes, at a 90\% confidence level. Error bars show standard errors over 100 random seeds. At $10^{8}$ viewers, roughly one day of traffic at Douyin, Doubly Robust DQ reliably achieves 80\% power for effect sizes as small as 0.15\%.}
  \label{fig:tiktok-power}
\end{figure}

\subsection{Robustness to Misspecification (RQ3)}

\paragraph{Doubly Robust} Finally, we consider what happens to each estimator when the treatment assignment probabilities are misspecified; i.e., unbeknownst to the algorithm, rather than the nominal probability $p = 1 / 2$, the environment actually executes an experiment with $p = 1 / 2 + 1 / 1000$. Robustness to such misspecification is a first-order concern at Douyin, where for example the intervention can cause a regression for a tiny proportion of streamers, who then no longer appear in viewers' feeds. \cite{kohavi2020trustworthy} describes a number of other scenarios in which such perturbations occur.

\cref{tbl:tiktok-perturb} shows results for all estimators in this setting. Here, we see that the bias and variance issues continue to plague Naive and OPE, as expected. However, Monte-Carlo DQ now becomes heavily biased as well, estimating a treatment effect of the wrong sign and of substantially larger magnitude than the actual treatment effect, despite the small change in the treatment probabilities. In contrast, the bias of Doubly Robust DQ is almost unchanged, remaining at most a few percent of the ATE. This is despite the fact that the regression estimator we use for $\hat{Q}^{\mathrm{Reg}}_{\pi_{1 / 2}}$ is extremely crude, and extremely poorly specified, speaking to the ease with which doubly robust estimators can be constucted, and to the outsize impact of doing so.

\begin{table}[htbp]
  \centering
\begin{tabular}[t]{rrrrrrrrr}
\toprule
Actual ATE (\%) & -0.01        &       & -0.17        &       & -0.50        &       & -1.89        &       \\
           & Mean         & SD    & Mean         & SD    & Mean         & SD    & Mean         & SD    \\
\midrule
DQ         & 4.98 (0.24)  & 2.20  & 4.49 (0.28)  & 2.56  & 4.43 (0.28)  & 2.52  & 2.70 (0.26)  & 2.37  \\
DQ-DR      & -0.06 (0.00) & 0.04  & -0.12 (0.00) & 0.03  & -0.48 (0.00) & 0.04  & -1.90 (0.00) & 0.04  \\
Naive      & 0.71 (0.02)  & 0.18  & 1.38 (0.02)  & 0.22  & 3.35 (0.02)  & 0.21  & 9.82 (0.02)  & 0.21  \\
Naive-DR   & 0.30 (0.00)  & 0.01  & 0.99 (0.00)  & 0.00  & 2.95 (0.00)  & 0.01  & 9.43 (0.00)  & 0.01  \\
OPE        & 1.22 (3.11)  & 28.80 & 5.94 (3.07)  & 28.13 & -0.40 (3.29) & 30.14 & 3.49 (2.59)  & 23.64 \\
OPE-DR     & -1.24 (3.32) & 26.81 & 4.57 (3.37)  & 26.52 & -1.69 (1.87) & 14.87 & -8.18 (2.31) & 18.73 \\
\bottomrule
\end{tabular}
  \caption{Estimation results for all estimators and various effects sizes ($-0.01\%$, $-0.17\%$, $-0.50\%$, $-1.89\%$), where the estimator uses $p=1 / 2$ but in fact the policy executes $p = 1 / 2 + 1 / 1000$. Naive and OPE estimators continue to display high bias / high variance respectively. Monte-Carlo DQ is now significantly biased, while Doubly Robust DQ remains essentially unperturbed.}
  \label{tbl:tiktok-perturb}
\end{table}


\bibliographystyle{plain}
\bibliography{bib}

\end{document}